\documentclass[journal,twoside]{IEEEtran}

\usepackage{amsmath,amssymb,amsfonts,mathtools}
\usepackage{amsthm}
\usepackage{bm}
\usepackage{booktabs}
\usepackage{graphicx}
\usepackage{cite}
\usepackage{url}
\usepackage{xcolor}
\usepackage{float}
\usepackage{comment}

\definecolor{revisionblue}{RGB}{0,0,0}
\newif\ifshowrevisions
\showrevisionstrue 

\newcommand{\revcolor}{\ifshowrevisions\color{revisionblue}\fi}

\makeatletter
\newcommand{\revisionbibcolor}[1]{\normalcolor\ifshowrevisions
  \ifcsname revisionbib@#1\endcsname\color{revisionblue}\fi\fi}
\let\revision@bibitem\@bibitem
\def\@bibitem#1{\revisionbibcolor{#1}\revision@bibitem{#1}}
\let\revision@lbibitem\@lbibitem
\def\@lbibitem[#1]#2{\revisionbibcolor{#2}\revision@lbibitem[#1]{#2}}
\@namedef{revisionbib@Heiss2012}{}
\@namedef{revisionbib@Dobson2001}{}
\@namedef{revisionbib@Weerts2018}{}
\@namedef{revisionbib@Fan2026Feedback}{}
\makeatother

\usepackage{enumitem}
\usepackage{array}

\newtheorem{theorem}{Theorem}
\newtheorem{lemma}{Lemma}
\newtheorem{proposition}{Proposition}
\newtheorem{corollary}{Corollary}
\newtheorem{definition}{Definition}

\newcommand{\C}{\mathbb{C}}
\newcommand{\eps}{\varepsilon}
\newcommand{\spec}{\sigma}

\newcommand{\Ylocal}{\mathcal{Y}_{\mathrm{local}}}
\newcommand{\YPMU}{\mathcal{Y}_{\mathrm{PMU}}}
\newcommand{\Yrich}{\mathcal{Y}_{\mathrm{rich}}}
\newcolumntype{P}[1]{>{\raggedright\arraybackslash}p{#1}}

\title{Source--Grid Coupling in Power System Oscillations: Observational Equivalence, Feedback Detectability, and Modal Interaction}

\newcommand{\AuthorNames}{}

\author{\AuthorNames}

\author{Kai~Sun,~\IEEEmembership{Fellow,~IEEE} and Bin~Wang,~\IEEEmembership{Senior Member,~IEEE}%
\thanks{K. Sun is with the Min H. Kao Department of Electrical Engineering and Computer Science, The University of Tennessee, Knoxville, TN 37996 USA (e-mail: kaisun@utk.edu). B. Wang is with ISO New England (e-mail: bwang@iso-ne.com).}}

\newcommand{\GridReal}{-0.094083}
\newcommand{\GridImag}{4.416645}
\newcommand{\GridFrequency}{0.70}
\newcommand{\GridDamping}{2.1}

\newcommand{\TwoModeProduct}{(5.3626+j6.4589)\times10^{-4}\,\mathrm{s}^{-2}}

\newcommand{\TwoModeRange}{1}

\newcommand{\LocalShortEpsilon}{1}
\newcommand{\LocalShortPower}{8.5}

\newcommand{\SmallestQualifiedCoupling}{1}
\newcommand{\FrequencyFalseNatural}{100}
\newcommand{\DetuneFarSNR}{11.25}
\newcommand{\DetuneFarPower}{73.5}
\newcommand{\DetuneNearSNR}{16.69}
\newcommand{\DetuneNearPower}{100.0}

\newcommand{\IDFiveSelection}{100.0}
\newcommand{\IDFiveEstimate}{1.102}
\newcommand{\IDFiveMSE}{1.220}
\newcommand{\IDFiveOracleMSE}{1.000}

\begin{document}
\maketitle

\begin{abstract}
The conventional distinction between forced and natural oscillations treats the oscillation source as external to a fixed power-system boundary. This paper develops a coupling-aware dynamical interpretation by explicitly accounting for the directionality of source–grid coupling. It examines when prescribed forcing and a one-way autonomous source are observationally equivalent, and when grid measurements can distinguish the presence of grid-to-source feedback. Matched interface signals yield exact equivalence under the same deterministic grid initial-value problem and measurement map. Bidirectional feedback enables source–grid eigenvalue interaction in the augmented autonomous model, while a reduced two-mode model shows that the interaction depends jointly on loop coupling and complex modal separation. The resulting framework establishes source–grid coupling directionality as a feedback-based dynamical boundary and clarifies when feedback can be inferred from grid-side measurements and why weak feedback may remain unresolved. Simulation studies on a two-area system demonstrate feedback detection within a calibrated source family, while source-model mismatch can be confounded with feedback. A WECC 243-bus system study demonstrates feasibility on a large system and shows that feedback detection need not require exact candidate selection.
\end{abstract}


\section{Introduction}
\IEEEPARstart{S}{ustained} oscillations are a continuing reliability concern in interconnected power systems. A conventional diagnosis first asks whether an observed component is a poorly damped natural mode or a response to periodic forcing. The distinction matters operationally: modal damping problems motivate controller retuning or wide-area damping actions, whereas a forced oscillation calls for locating, isolating, or modifying the responsible device. Synchrophasor measurements have therefore motivated methods for forced-oscillation detection, discrimination, localization, and suppression \cite{Wang2017Location,IEEE_TR110_2023,FollumPierre2016,Ye2017,Trudnowski2020,Cai2026}. Source-location methods based on dissipating energy, effective generator impedance, unknown-input observers, and data-driven residuals provide complementary information about where oscillatory power is injected or sustained \cite{Jha2019,Chevalier2018,Huang2020,JiangSWT2024,Tao2024}. Near resonance, however, forced responses may exhibit the same frequency and spatial signatures as a lightly damped local grid mode, complicating diagnosis \cite{Sarmadi2016,Huang2020,XieTrudnowski2017}.

An oscillation source is conventionally modeled as an exogenous periodic input to the grid model. This modeling choice is useful, but it leaves the physical dynamics of the source outside the system boundary. A malfunctioning controller, converter, cyclic load, turbine-governor component, or other oscillatory device generally has internal states and may itself be affected by grid-side terminal voltage, frequency, or power. {Including these dynamics permits analysis of source--grid interaction; association with equilibrium modes additionally requires a valid small-signal description.} Thus, when such devices are treated as oscillation sources, they need not be ideally exogenous. The diagnosis then involves not only oscillation frequency and mode shape, but also: 1) \emph{the directionality of source--grid coupling} and 2) \emph{whether that directionality can be inferred from available grid-side measurements}.

This paper systematically answers those two questions under a general grid-source coupling model:
\begin{align}
    \dot x &= f(x,z), \label{eq:gx}\\
    \dot z &= g(z,x), \label{eq:gz}\\
    y &= h(x,u_s). \label{eq:gy}
\end{align} 
where \(x\), $z$ and $y$ respectively denote grid dynamic states, hidden source states, and grid-side measurements. {The source acts through interface $u_s$; dependence of a grid output on this interface does not mean that an internal source state is measured. If $g(z,x)$ is independent of $x$ on the domain of interest, the
source acts one-way on the grid. Otherwise, with source-to-grid
influence present, the source and grid are bidirectionally coupled;
a zero grid-to-source Jacobian at one equilibrium does not by itself
exclude nonlinear feedback.} The central question is not whether a periodic signal can be written as an autonomous state-space model---that construction is classical \cite{FrancisWonham1976}---but when the prescribed forcing is observationally equivalent to a one-way autonomous source (Q1), and when grid measurements are able to distinguish whether or not a source receives grid-to-source feedback (Q2). An engineer may need to include return influence in a dynamic model before a unique internal source model can be recovered.

The main contributions of the paper are as follows.
\begin{enumerate}[leftmargin=*,label=\arabic*)]

\item Three source classes are defined: prescribed exogenous forcing
(H0), a one-way autonomous source (H1), and a bidirectionally coupled
source (H2), separating source representation from coupling direction.

\item Matched H0 and H1 realizations are proved to be observationally
equivalent from grid-side measurements when they generate the same
interface signal; distinguishing them therefore requires source-side
information, intervention, or additional assumptions.

\item The paper establishes when grid-side measurements can detect grid-to-source feedback under uncertain source parameters, measurement noise, and finite data, and shows that feedback detection is distinct from—and generally easier than—full source and parameter identification.

\item A simple modal interpretation of source–grid feedback is developed using perturbation theory to show how bidirectional coupling shifts and mixes source and grid modes, depending on coupling strength and modal separation.

\item The proposed framework is validated on both small and large systems,
covering H0/H1 equivalence, conditional H1/H2 feedback detection,
weak-feedback limitations, model mismatch, and the distinction between
feedback detection and parameter recovery.
\end{enumerate}

Taken together, these contributions establish a coupling-aware framework
for interpreting sustained oscillations from grid-side measurements. The
framework separates two questions that are often conflated: how the
oscillation source is represented relative to the grid boundary, and
whether the grid dynamically feeds back into that source. It further
clarifies the limits of grid-side inference and distinguishes feedback
detection from hidden-source parameter recovery. The spectral tools used
to characterize bidirectional coupling are established results
\cite{Heiss2012,Dobson2001,Fan2026Feedback}; this paper connects them to
the forced--natural oscillation boundary. In particular, matched H0/H1
realizations can be observationally equivalent, whereas H1/H2
discrimination is possible only under appropriate model, measurement,
and numerical conditions, consistent with model-conditioned
identifiability \cite{Weerts2018}. In the rest of the paper, Sections II--V develop the source--grid coupling framework, establish feedback distinguishability, examine its spectral effects, and formulate a coupling-aware diagnostic interpretation; Sections VI and VII present comprehensive tests on the two-area system and a demonstration on the WECC system; Section VIII concludes the paper.

\section{Proposed Directional Source–Grid Coupling Framework}
\subsection{Conventional Fixed-Boundary Forced/Natural Formulation}
A power-system dynamic model can be written as differential-algebraic equations
\begin{align}
    \dot x_g &= f_g(x_g,v,u_s), \label{eq:dae1}\\
    0 &= g_g(x_g,v,u_s), \label{eq:dae2}\\
    y &= h_g(x_g,v), \label{eq:dae3}
\end{align}
where \(x_g\) contains generator, exciter, governor, converter, and control states; \(v\) contains algebraic bus variables; and \(u_s\) is the source-interface signal. Linearization and elimination of regular algebraic variables yield
\begin{equation}
    \delta\dot x=A\delta x+B_u\delta u_s,\qquad
    \delta y=H\delta x+J\delta u_s. \label{eq:reduced_grid}
\end{equation}
where $\delta x$, $\delta u_s$, and $\delta y$ denote small deviations of the state, source-interface signal, and output from their equilibrium values, respectively.
With \(\delta u_s=0\), the eigenvalues of \(A\) define the conventional small-signal modes \cite{Kundur1994}. With a periodic \(\delta u_s\), the steady response is governed by the transfer matrix from the injection channel to the measured outputs. A large forced response near a lightly damped pole is therefore expected even when the forcing amplitude is small \cite{Sarmadi2016}.

Compared to the complete model in equations (\ref{eq:gx})-(\ref{eq:gy}), the reduced model does not specify whether \(u_s\) is generated by a clocked reference, a hidden oscillator, or a device that is dynamically affected by the grid. That distinction motivates the following source hypotheses.

\subsection{Three Source Hypotheses and the Feedback Boundary}
To distinguish different source--grid interaction structures, consider these three hypotheses according to the direction of dynamical coupling between the source and the grid.

\vspace{-0.1in}
\begin{table}[H]
\centering
\caption{Source Models and Their Grid-Side Meaning}
\label{tab:sourceclasses}
\footnotesize
\setlength{\tabcolsep}{3pt}
\begin{tabular}{@{}P{0.07\columnwidth}P{0.38\columnwidth}P{0.43\columnwidth}@{}}
\toprule
Class & Source model and feedback & Grid-side meaning\\
\midrule
H0 & Explicit \(u_s(t)\); source dynamics outside the model & Conventional forced input\\
H1 & \(\dot z=g_0(z)\); no grid-to-source feedback & Hidden generator of an interface signal\\
H2 & \(\dot z=g_0(z)+\eps q(z,x)\); feedback present & Closed source--grid dynamics\\
\bottomrule
\end{tabular}
\end{table}

\textbf{H0: Prescribed forcing.} The source is represented only by a known or unknown time signal,
\begin{equation}
    \dot x=f_0(x,u_s(t)),\qquad y={h(x,u_s)}. \label{eq:H0}
\end{equation}
A common example is \(u_s(t)=a\sin(\omega_f t+\phi)\).

\textbf{H1: One-way autonomous source.} The source has hidden states but no grid-to-source feedback,
\begin{align}
    \dot x&=f(x,z),\label{eq:H1x}\\
    \dot z&=g_0(z).\label{eq:H1z}
\end{align}
The source drives the grid, while its trajectory is unaffected by grid variables.

\textbf{H2: Bidirectionally coupled source.} The source receives a feedback signal from the grid,
\begin{align}
    \dot x&=f(x,z),\label{eq:H2x}\\
    \dot z&=g_0(z)+\eps q(z,x),\label{eq:H2z}
\end{align}
where $\eps\geq 0$ is a dimensionless continuation parameter that scales
a fixed feedback law with specified physical units. The feedback sign and
direction are determined by $q$, or by a calibrated matrix $C$ after
linearization. A positive $\eps$ corresponds to H2 provided that the
specified feedback path is nonzero, whereas $\eps=0$ recovers the one-way
boundary. Note that the decomposition of the grid-to-source feedback into the scalar
coupling parameter $\eps$ and the feedback map $q$ is introduced primarily
for continuation and sensitivity analyses. Physically, the feedback is
determined by the effective feedback law $\eps q$.

The proposed framework addresses two primary diagnostic questions: Can different one-way source realizations, particularly H0 and H1, be distinguished from grid-side measurements (Q1)? Can grid-to-source feedback be detected, thereby distinguishing H2 from the one-way H1 boundary (Q2)?  
Considering that statistical inference studies can compare only specified candidate source families, for a stated source
parameterization, define:
\begin{equation}
    \mathcal{M}_{0}^{\mathrm{cal}}:\ \eps=0,
    \qquad
    \mathcal{M}_{+}^{\mathrm{cal}}:\ \eps>0 .
    \label{eq:calibrated_model_families}
\end{equation}
Here, both families share the stipulated source structure and admissible
nuisance-parameter ranges. Thus, $\mathcal{M}_{0}^{\mathrm{cal}}$ is a
restricted zero-feedback candidate family rather than the unrestricted
set H0$\cup$H1; in particular, it does not include an arbitrary prescribed H0 waveform.

\emph{Remark:} After augmenting the grid states with the source states,
as shown in \eqref{eq:gx} and \eqref{eq:gz}, both H1 and H2 form
autonomous dynamical systems. Their distinction is therefore not
autonomy, but whether the grid exerts a return influence on the source
dynamics: this influence is absent in H1 and present in H2.
Autonomy should also not be confused with absence of energy supply.
A self-sustained device can be autonomous and dissipative while
maintaining a limit cycle using an internal electrical or mechanical
energy supply \cite{Strogatz2015}, such as the Van der Pol-type
oscillation source introduced in Section~V-B.
Similarly, a device may draw average power from the grid while remaining
dynamically insensitive to the grid variables associated with the
oscillatory feedback path. Thus, an energetic connection does not by itself imply dynamical
grid-to-source feedback in \eqref{eq:H2z}, nor does the existence of
such feedback guarantee that it can be identified from grid-side
measurements.

\section{Observational Equivalence and Conditional Feedback Distinguishability}
\subsection{H0--H1 Observational Equivalence}
\begin{definition}[Grid-side observational equivalence]
Two hidden-source models are grid-side observationally equivalent on \([0,T]\) if, when connected to the same deterministic grid model with the same initial grid state, they produce identical measured trajectories \(y(t)\) on that interval.
\end{definition}

\begin{theorem}[Interface-signal equivalence]\label{thm:interface}
Assume the grid model \(\dot x=f_0(x,u_s)\), \(y={h(x,u_s)}\), has a unique solution for each admissible interface signal and initial state. {The two interconnections use this same deterministic grid model, connection interface, other applied inputs, grid initial condition, and measurement map.} If two source realizations generate the same \(u_s(t)\) on \([0,T]\), then they generate the same \(x(t)\) and \(y(t)\) on \([0,T]\). No measurement confined to the grid trajectory can distinguish the two source realizations.
\end{theorem}
\begin{proof}
Both interconnections solve the same initial-value problem (IVP) with the same input and initial condition. Uniqueness gives identical grid-state trajectories, and the output map then gives identical measurements.
\end{proof}

\emph{Interpretation:}
This mathematical statement follows directly from uniqueness of the grid IVP; its significance is the resulting inference limitation: two hidden sources producing the same interface signal are indistinguishable from grid-side measurements alone, even with complete grid-state measurement, unless additional source-side information, intervention, or structural assumptions are available.

Equality of interface signals is sufficient but is not claimed to be necessary: output equality alone need not determine an input without additional observability or input-reconstruction assumptions.

For a sinusoidal source, define
\begin{equation}
    \dot z=S_\omega z,\quad
    S_\omega=\begin{bmatrix}0&\omega\\-\omega&0\end{bmatrix},\quad
    u_s=Lz. \label{eq:exosystem}
\end{equation}
Appropriate \(z(0)\) and \(L\) reproduce any \(a\sin(\omega t+\phi)\), consistent with the classical autonomous-signal-generator construction \cite{FrancisWonham1976}.

\begin{corollary}[Rigid sinusoid versus hidden oscillator]\label{cor:h0h1}
A prescribed sinusoidal input (H0) and an unobserved one-way harmonic oscillator (H1) are exactly grid-side observationally equivalent when they generate the same interface signal, even if the full grid state is measured.
\end{corollary}

\begin{corollary}[Single-trajectory replay]\label{cor:replay}
Suppose an H2 interconnection drives the grid only through the interface variable in \(\dot x=f_0(x,u_s)\), \(y={h(x,u_s)}\), and generates an admissible realized signal \(u_s^*(t)\) on \([0,T]\). Assume unique grid solutions and the same initial grid state. An H0 model prescribing \(u_s(t)=u_s^*(t)\) reproduces the identical grid trajectory and output on that interval.
\end{corollary}
\begin{proof}
Apply Theorem~\ref{thm:interface} to the realized signal and its prescribed replay.
\end{proof}
This equivalence is limited to a single realized trajectory and need not
hold under different initial conditions or interventions. The H0 replay
need not be a fixed-frequency sinusoid or realizable by the stipulated H1
model. Across experiments, a common source model with shared parameters can
constrain replay and make grid-to-source feedback distinguishable, whereas
allowing an unrestricted H0 waveform for each record preserves the
ambiguity. The later H1/H2 studies therefore impose specified source
families.

This replay ambiguity limits inference from grid-side measurements alone.
A PMU-based method may identify an oscillatory interface, localize a bus or
device, or reject a baseline modal model, but cannot determine whether the
same one-way waveform arose from a prescribed external input, an internal
autonomous oscillator, or another hidden realization. Resolving that
distinction requires source-side measurements, interventions, or additional
structural assumptions.

\vspace{-0.2in}
\subsection{{Conditional Feedback Evidence and Local Sensitivities}}
At a specified equilibrium, let the linearized grid dynamics and output
depend on the source-interface signal $u_s$ as
\begin{align}
    \dot x &= Ax+B_u u_s,\\
    y &= Hx+J u_s,
\end{align}
where $B_u$ maps the source-interface signal into the grid dynamics and
$J$ represents any direct dependence of the measured outputs on that
interface signal. Let the source output be $u_s=Lz$, where $L$ maps the source state $z$ to the grid interface. For H2, the
linearized source dynamics are $\dot z=Dz+\eps Cx.$
Substituting $u_s=Lz$ gives
\begin{align}
    \dot x &= Ax+Bz, \qquad B=B_uL, \label{eq:linx}\\
    \dot z &= Dz+\eps Cx, \label{eq:linz}\\
    y &= Hx+JLz, \label{eq:liny}
\end{align}
with
\begin{equation}
    M(\eps)=
    \begin{bmatrix}
        A & B\\
        \eps C & D
    \end{bmatrix}.
    \label{eq:M}
\end{equation}
Thus, $B$ represents source-to-grid influence, whereas $\eps C$
represents grid-to-source feedback. The corresponding augmented output
map is $H_{\rm aug}=[H\;JL]$. The case $J=0$ means that the measured
outputs have no direct feedthrough from $u_s$; it is a special case
rather than a general assumption for voltage or power measurements.

\emph{Remark (Coupling normalization):}\label{rem:normalization}
If $\eps$ and $C$ enter the model only through $\eps C$, their
separate scales cannot be inferred without additional information.
We therefore fix the feedback normalization before estimating $\eps$.
This convention removes the factorization ambiguity but does not
itself establish identifiability of the feedback from grid-side data. Ambiguity in the factors does not by itself imply ambiguity about whether their effective product is nonzero.

For the nonlinear feedback law in \eqref{eq:H2z}, the local matrix is
$C=\partial q/\partial x\vert_{(z_*,x_*)}$ at a reference equilibrium
$(z_*,x_*)$. At that fixed equilibrium, $C$ does not change with
disturbance amplitude. Along a nonlinear trajectory, however,
$C(t)=\partial q/\partial x\vert_{(z(t),x(t))}$ may vary, and a constant
effective gain fitted to a finite-amplitude record may depend on the
record. The continuation model \eqref{eq:M} holds $A$, $B$, $C$, and
$D$ fixed while varying $\eps$; relinearization at a different
operating point can change these matrices. The present study does
not identify a general state-dependent feedback law.

After fixing the normalization, the inference question is whether
feedback can be distinguished from uncertainty in the assumed source
dynamics \cite{BellmanAstrom1970}. First consider the case with
only $\eps$ unknown. Let $X=[x^\top,z^\top]^\top$,
$Y=H_{\rm aug}X$, and
\begin{equation}
    S_\eps(t)=\frac{\partial Y(t;\eps)}{\partial\eps}. \label{eq:sens}
\end{equation}
The sensitivity state satisfies
\begin{align}
    \dot X_\eps&=M(\eps)X_\eps+M_\eps X,\label{eq:sensstate}\\
    S_\eps&={H_{\rm aug}X_\eps},\qquad X_\eps(0)=0,\label{eq:sensout}
\end{align}
where \(M_\eps=\left[\begin{smallmatrix}0&0\\C&0\end{smallmatrix}\right]\).
These equations assume an $\eps$-independent output and initialization,
as in the calibrated scalar-feedback runs. A parameter-dependent output
adds $(\partial_\eps H_{\rm aug})X$, and a parameter-dependent initial
state requires its initial sensitivity.

\begin{proposition}[Sufficient local identifiability]\label{prop:localid}
Assume all matrices, coordinates, and initial conditions are fixed except the scalar \(\eps\), and \(Y(t;\eps)\) is continuously differentiable. If one measured component satisfies \(\partial Y_k(t_*;\eps)/\partial\eps\neq0\) at \(\eps=\eps_0\), then \(\eps\) is locally identifiable at \(\eps_0\) from the noiseless trajectory.
\end{proposition}
\begin{proof}
The scalar map $\eps\mapsto Y_k(t_*;\eps)$ has nonzero derivative and
is therefore locally one-to-one by the inverse function theorem; at the
boundary $\eps_0=0$, the same conclusion holds on the one-sided
admissible domain $\eps\geq0$ by continuity of the derivative.
\end{proof}

The proposition is sufficient, not necessary: vanishing first-order sensitivity may coexist with higher-order parameter dependence. More importantly, if source order, \(D\), \(C\), and source initial states are unknown, hidden-state similarities and compensating parameter changes can preserve the same grid output. The numerical study must therefore state which source structure is assumed before estimating \(\eps\).

When other source parameters are also unknown, nonzero sensitivity to
$\eps$ alone does not guarantee that $\eps$ can be resolved separately,
because changes in those parameters may produce similar measurement
changes. Joint local resolution must therefore be assessed from the
combined parameter-sensitivity matrix. This issue is examined explicitly
in the multi-parameter studies of Section~VI-D.

For a sampled and stacked measurement vector with noise covariance
$\Sigma$, define the whitened parameter sensitivity as
$\bar{s}_{\theta}=\Sigma^{-1/2}s_{\theta}$. Let $\bar{S}_{\nu}=
\begin{bmatrix}
\bar{s}_{\nu_1} & \cdots & \bar{s}_{\nu_p}
\end{bmatrix}
$
collect the corresponding sensitivities to the nuisance parameters
$\nu=[\nu_1,\ldots,\nu_p]^\top$. Define
\begin{equation}
r_{\eps}
=
\left(I-P_{\bar{S}_{\nu}}\right)\bar{s}_{\eps},
\qquad
P_{\bar{S}_{\nu}}
=
\bar{S}_{\nu}\bar{S}_{\nu}^{\dagger},
\label{eq:projected_feedback_sensitivity}
\end{equation}
where $(\cdot)^\dagger$ denotes the Moore--Penrose pseudoinverse.

\emph{Remark:}
If $r_{\eps}\neq 0$, the first-order measurement change produced by
$\eps$ is not contained in the local sensitivity subspace generated by
the nuisance parameters. Hence, within the assumed model and
measurement map, feedback is locally distinguishable to first order
from perturbations of those nuisance parameters. If $r_{\eps}=0$, the
feedback effect is first-order confounded with nuisance-parameter
variation; higher-order information may nevertheless distinguish the
models. This is a local first-order condition and does not imply global
or structural identifiability.

\vspace{-0.2in}
\subsection{Finite-Data Feedback Detectability}
For sampled PMU data
\begin{equation}
    y_k=Y(t_k;\eps)+\eta_k,\qquad \eta_k\overset{\mathrm{iid}}{\sim}\mathcal N(0,R),
\end{equation}
a local scalar Fisher information is calculated by
\begin{equation}
    \mathcal I_\eps=\sum_{k=1}^{N}S_\eps(t_k)^\top R^{-1}S_\eps(t_k). \label{eq:fisher}
\end{equation}
Under standard regularity assumptions, \(\operatorname{var}(\widehat\eps)\ge \mathcal I_\eps^{-1}\) for an unbiased estimator \cite{Ljung1999}. Structural identifiability therefore does not imply practical detection. The test
\begin{equation}
    \mathcal{T}_0:\mathcal{M}_{0}^{\mathrm{cal}}(\eps=0),\qquad \mathcal{T}_+:\mathcal{M}_{+}^{\mathrm{cal}}(\eps>0) \label{eq:hyp}
\end{equation}
compares the restricted zero-feedback family \(\mathcal M_{0}^{\rm cal}\) with its feedback-augmented extension \(\mathcal M_{+}^{\rm cal}\) under the stated source parameterization. Its false-selection probability at \(\eps=0\) and detection probability at \(\eps>0\) depend on SNR, window length, sampling, sensor map, source amplitude, detuning, and model uncertainty. Failure to reject \(\mathcal{T}_0\) does not establish absence of physical feedback. There is no universal boundary between weak and strong coupling.

When all source parameters other than $\eps$ are fixed, a numerically
resolved difference between the $\eps=0$ and $\eps>0$ responses provides
evidence for feedback. If other source parameters are also unknown,
both the zero-feedback and feedback models must be fitted over the same
admissible ranges of those parameters. A natural comparison is the
reduction in weighted fitting error,
\begin{equation}
\Lambda = J(0)-J(\hat{\eps}), \label{eq:Lambda}    
\end{equation}
where $J(0)$ is the best fit achievable with zero feedback and
$J(\hat{\eps})$ is the best fit of the feedback-augmented calibrated family when $\eps\geq0$ is also allowed to vary.
The weighting accounts for the different noise levels and units of the
measured channels. A positive improvement can support feedback detection,
but need not uniquely determine the other source parameters. Because the
response varies continuously with $\eps$, sufficiently weak feedback may
remain indistinguishable from zero feedback with finite noisy data.

\section{Spectral Consequences of Bidirectional Source–Grid Feedback}
This section specializes standard block-matrix perturbation and two-mode interaction results in \cite{StewartSun1990,Heiss2012,Dobson2001} to source--grid feedback.

\vspace{-0.2in}
\subsection{One-Way Spectrum and Feedback Sensitivity}
At zero feedback, \(M(0)\) is block upper triangular.
\begin{lemma}[One-way spectrum]\label{lem:spectrum}
For \(\eps=0\),
\begin{equation}
    \spec(M(0))=\spec(A)\cup\spec(D), \label{eq:specunion}
\end{equation}
counting algebraic multiplicities.
\end{lemma}
\begin{proof}
The characteristic determinant factors as \(\det(\lambda I-A)\det(\lambda I-D)\).
\end{proof}

The presence of \(\spec(D)\) in an augmented autonomous realization does not make those eigenvalues conventional modes of the original grid matrix \(A\), nor does it remove the H0--H1 observational equivalence. Let \(\lambda\in\spec(A)\setminus\spec(D)\) be simple, with right and left eigenvectors \(v_A,w_A\). Standard eigenvalue perturbation gives \cite{StewartSun1990}:
\begin{proposition}[Grid-mode feedback sensitivity]\label{prop:gridshift}
\begin{equation}
\left.\frac{d\lambda(\eps)}{d\eps}\right|_{0}
=\frac{w_A^*B(\lambda I-D)^{-1}Cv_A}{w_A^*v_A}. \label{eq:gridshift}
\end{equation}
For a simple source eigenvalue \(\mu\in\spec(D)\setminus\spec(A)\),
\begin{equation}
\left.\frac{d\mu(\eps)}{d\eps}\right|_{0}
=\frac{w_D^*C(\mu I-A)^{-1}Bv_D}{w_D^*v_D}. \label{eq:sourceshift}
\end{equation}
\end{proposition}

The resolvents in \eqref{eq:gridshift}--\eqref{eq:sourceshift} show that the effect of small feedback depends on source--grid modal separation and nonnormal modal geometry, not on \(\eps\) alone.

\subsection{{Two-Mode Source–Grid Interaction}}
Retaining one grid mode and one source mode gives the approximation
\begin{equation}
    M_r(\eps)=
    \begin{bmatrix}
        \lambda_g & b\\
        \eps c & \lambda_s
    \end{bmatrix},
    \label{eq:Mr}
\end{equation}
where, under compatible biorthogonal normalization,
$b=w_g^*Bv_s$ and $c=w_s^*Cv_g$, so that $bc$ is the effective
source--grid loop product for the selected mode pair. This two-mode
model neglects the influence of other system modes and is therefore
valid only when their effect on the selected pair is sufficiently small.
Theorem~2 below is exact for the reduced matrix \eqref{eq:Mr}, not in
general for the full system.

\begin{theorem}[Two-mode source--grid splitting]\label{thm:split}
The eigenvalues of \eqref{eq:Mr} are
\begin{equation}
\lambda_{\pm}=\frac{\lambda_g+\lambda_s}{2}
\pm\sqrt{\left(\frac{\lambda_g-\lambda_s}{2}\right)^2+\eps bc}. \label{eq:lpm}
\end{equation}
\end{theorem}
\begin{proof}
The characteristic equation is \((\lambda-\lambda_g)(\lambda-\lambda_s)-\eps bc=0\); applying the quadratic formula yields \eqref{eq:lpm}.
\end{proof}

For \(\Delta=\lambda_g-\lambda_s\neq0\), define
\begin{equation}
    \chi=\frac{|\eps bc|}{|\Delta|^2}. \label{eq:chi}
\end{equation}
When \(\chi\ll1\), the branch connected to \(\lambda_g\) satisfies
\begin{equation}
    \lambda_g(\eps)=\lambda_g+\frac{\eps bc}{\Delta}
    +\mathcal O\!\left(\frac{\eps^2|bc|^2}{|\Delta|^3}\right). \label{eq:far}
\end{equation}
Within a valid two-mode description, $\chi=O(1)$ calls for treating the eigenvalues and eigenvectors jointly. It does not quantify physical source/grid energy shares. At exact coincidence, \(\lambda_g=\lambda_s\), and nonzero \(bc\), the splitting is proportional to \(\sqrt{\eps bc}\), so an expansion analytic in \(\eps\) fails; \(\chi\) is undefined at coincidence.

Although $b$ and $c$ individually depend on the normalization of the
retained modal vectors, their product $bc$, and hence $\chi$, is invariant
under reciprocal modal scaling. This invariance does not make $\chi$ a
universal coupling index: it characterizes only the selected mode pair
within a valid two-mode approximation, and its interpretation may change
if important modes are omitted. Open-loop resonance under H0 or H1
amplifies the response without changing the grid spectrum, whereas H2
feedback can move the interconnected eigenvalues and rotate their
eigenvectors.

\section{From Forced/Natural Labels to Coupling-Aware Interpretation}
\subsection{A Coupling-Aware Diagnostic Hierarchy}

The preceding results show that a sustained oscillation should not be
classified from frequency or mode-shape similarity alone. In practice,
different questions must be separated: whether an oscillatory component
is present, whether it is consistent with a known grid mode, where it
enters the grid, and whether the grid dynamically feeds back into the
source. These questions require different models and measurements, and
confounding them can lead to an incorrect forced/natural interpretation.
A hierarchy of coupling-aware interpretation is summarized in Table~\ref{tab:interpretation} to interpret a given observed oscillation and specify the evidence needed for the interpretation.

A practical diagnostic sequence is: (i) detect and characterize the
sustained component; (ii) compare it with the best available grid modal
model; (iii) localize the injection/interface using existing methods;
(iv) specify admissible one-way and feedback source models, including
uncertainty in source parameters; (v) test for grid-to-source feedback
using numerically resolved response differences and calibrated
measurement-model comparisons; and (vi) report detection confidence as
a function of SNR, observation window, and detuning. A spectral peak
alone should not be used to infer hidden source physics. Detecting return
influence supports a coupled H2 interpretation, but does not by itself
identify hidden-source participation factors or all modes of the
augmented system. Modal analysis remains useful for H1 forced responses,
whereas limit cycles and strongly nonlinear behavior require appropriate
nonlinear analysis.

\begin{table}[!t]
\centering
\caption{Coupling-Aware Interpretation of an Observed Oscillation}
\label{tab:interpretation}
\scriptsize
\setlength{\tabcolsep}{2pt}
{\renewcommand{\arraystretch}{1.4}
\begin{tabular}{@{}P{0.29\columnwidth}P{0.68\columnwidth}@{}}
\toprule
Interpretation & Meaning and implication\\
\midrule
Grid natural mode &
Mode of the baseline grid model; compare frequency, damping, and mode shape.\\

One-way-equivalent forcing &
H0 or H1 with the same interface waveform; grid-side data alone cannot
distinguish their source realizations.\\

Coupled-source oscillation &
H2 with detected grid-to-source feedback under the stated source-model
and measurement assumptions; source-parameter recovery is separate.\\

Potential strong modal interaction &
$\chi=O(1)$ in a valid two-mode approximation; the selected grid and
source modes may need joint treatment.\\

Nonlinear self-sustained oscillation &
Limit-cycle, frequency-pulling, or locking behavior; equilibrium
eigenanalysis alone is insufficient.\\
\bottomrule
\end{tabular}
}
\end{table}

\vspace{-0.2in}
\subsection{Nonlinear Self-Sustained Sources}
The preceding analysis uses linearized source and grid models to establish
the H1--H2 feedback boundary and its modal consequences. However, an
autonomous oscillation source need not be linear or oscillate about an
equilibrium. This subsection extends the proposed framework to nonlinear self-sustained oscillation sources that can sustain a stable limit cycle without a prescribed periodic input. The following Van der Pol-type source
provides a concrete nonlinear example of H1 and is used later in the
validation studies to test the proposed coupling-aware framework:
\begin{align}
    \dot z_1&=z_2,\label{eq:vdp1}\\
    \dot z_2&=\mu(1-z_1^2)z_2-\omega_0^2z_1+\eps q(x),\label{eq:vdp2}\\
    u_s&=k_z z_1.\label{eq:vdp3}
\end{align}
For $\mu>0$, $\omega_0>0$, and a nonzero source initial state, the source
is autonomous and approaches a self-sustained oscillation when
$\eps=0$, corresponding to a nonlinear H1 source. When $\eps>0$, the term
$q(x)$ introduces grid-to-source feedback and the source becomes H2. Such
feedback may produce frequency pulling, phase locking, entrainment, or
amplitude changes \cite{Strogatz2015,Kuramoto1984}. This example illustrates
that sustained oscillations can arise from nonlinear limit-cycle dynamics,
so equilibrium eigenanalysis alone is not always sufficient for
classification. Moreover, synchrony alone does not establish feedback,
because a damped grid may also follow a one-way oscillatory source.

\vspace{-0.2in}
\subsection{Scope and Limits of Coupling-Aware Diagnosis}

The coupling-aware interpretation has both modeling and inference limits.
First, representing an exogenous signal generator as additional autonomous
states does not by itself make the source physically internal to the power
system; the choice of system boundary remains a modeling assumption.
Explicitly time-varying or stochastic sources may also fall outside the
finite-dimensional deterministic source models considered here. Second,
even when the H1--H2 framework applies, grid-to-source feedback may be too
weak or poorly observed to distinguish from one-way forcing using finite
grid-side measurements. Thus, the framework distinguishes the dynamical
question of whether feedback is present from the inference question of
whether that feedback can be detected under the available model,
measurements, and numerical resolution.

\section{Validation on a Two-Area System}
\subsection{Grid-Source Modeling and Measurements}
Simulations were conducted by Powertech TSAT on its accompanying Kundur two-area system model, which contains 11 buses, 16 branches, and four GENROU generators (G1 and G2 in area 1 and G3 and G4 in area 2), each equipped with an ESAC4A exciter. G1's native ESAC4A exciter was replaced by a user-defined model (UDM) implementing the same exciter equations and limits, thereby allowing the exciter dynamics to be accessed and modified for the source--grid coupling experiments. In this validation study, the oscillation-source location is therefore known and fixed at the G1 AVR input summing junction within the UDM. The remainder of the baseline model is retained unchanged. Powertech SSAT was used to identify the dominant interarea mode as $\lambda_g=\GridReal+j\GridImag\,\mathrm{s}^{-1}$, or \GridFrequency~Hz with \GridDamping\% damping. The system also has two local modes at 1.159581~Hz and 1.193696~Hz with 10.5782\% damping and 10.4176 \% damping,  respectively.

The equations of the oscillation source are
\begin{align}
 \dot z_1&=z_2,\nonumber\\
 \dot z_2&=-\omega_s^2z_1-2\zeta_s\omega_sz_2+\eps K_0\Delta\omega_1,
 \quad u_s=k_z z_1,\label{eq:tested_source}
\end{align}
where $\Delta\omega_1=(f_1-60)/60$ pu is the normalized G1 rotor-speed deviation, with $f_1$ as the rotor frequency in Hz; it is distinct from the bus-voltage frequency used in the local observation vector, $K_0=1000\,\mathrm{s}^{-2}$, and $k_z=5\times10^{-4}$~pu per source unit. With $K_0$ fixed as the reference feedback gain, $\eps$ is a
dimensionless parameter that scales the feedback strength. The principal linear source has $\zeta_s=0$, $z(0)=(0,\omega_s)$, and both interface directions are enabled at $t=1$~s. The source phase is referenced to absolute simulation time. H0 uses a prescribed sine evaluated from an exact binary-step clock; H1 uses the autonomous harmonic equations at $\eps=0$.

Three levels of grid-side observation are considered: 1) The local configuration $\Ylocal=\{f_i,V_i,\theta_i,P_i,Q_i\}_{\mathrm{G1}}$ contains the frequency, voltage magnitude and angle and the active and reactive powers at the G1 terminal bus; 2) The PMU-like multi-generator configuration $\YPMU=\{f_i,V_i,\theta_i,P_i,Q_i\}_{\mathrm{G1-G4}}$ contains the same five channels at all four generator buses; 3) The rich configuration $\Yrich$ contains all available downstream grid-side quantities from the simulation results including 145 TSAT-exported channels, of which 137 satisfy the variance and numerical-resolution criteria and are retained in
$Y_{\rm rich}$; these are grid-side outputs rather than full native-state
measurements. No oscillation-source state enters any observation configuration. These increasingly rich configurations test dependence on available grid-side information.

The case study first addresses question Q1 (H0/H1 equivalence) at the one-way boundary and then
examines question Q2 (H1/H2 distinction) using specified H1/H2 response contrasts and finite noisy
measurements. Additional tests examine how uncertainty in source
parameters affects feedback detection and parameter resolution.
Nonlinear and model-mismatch cases are then used to assess the limits
of these conclusions.

\vspace{-0.2in}
\subsection{Q1: H0/H1 Equivalence}
\begin{figure}[!t]\centering
\includegraphics[width=0.7\columnwidth,height=1.3in]{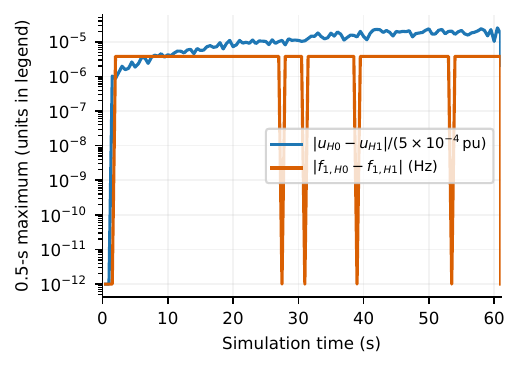}
\vspace{-0.2in}
\caption{Near-resonant H0/H1 discrepancy for the representative
$5\times10^{-4}$-pu, $1/128$-s run. Blue shows the normalized
source-interface error, and orange shows the G1 rotor-speed error in Hz.
The maximum source-interface error is $1.27\times10^{-8}$~pu, below the
$10^{-7}$-pu matching criterion.}\label{fig:equivalence}
\vspace{-0.2in}
\end{figure}

Consider the one-way boundary of the framework. H1 is evaluated with
$\eps=0$, so its autonomous source dynamics are unaffected by the grid,
while H0 is prescribed to generate the corresponding matched interface
signal. Theorem~\ref{thm:interface} predicts exact grid-side equivalence
for matched interface signals; the following comparisons therefore serve
as a numerical verification of this result.

The sources are chosen at 0.3~Hz and \GridFrequency~Hz to represent
off-resonant and near-resonant conditions, with interface amplitudes of
$5\times10^{-4}$ and $10^{-3}$~pu. A $10^{-7}$~pu interface-matching
criterion is imposed before comparing the grid responses. All retained
cases show close H0/H1 agreement across the considered observation
configurations, consistent with the matched-input equivalence of
Theorem~\ref{thm:interface}. Fig.~\ref{fig:equivalence} shows a
representative near-resonant case.

\begin{figure}[!t]\centering
\includegraphics[width=0.7\columnwidth,height=2.in]{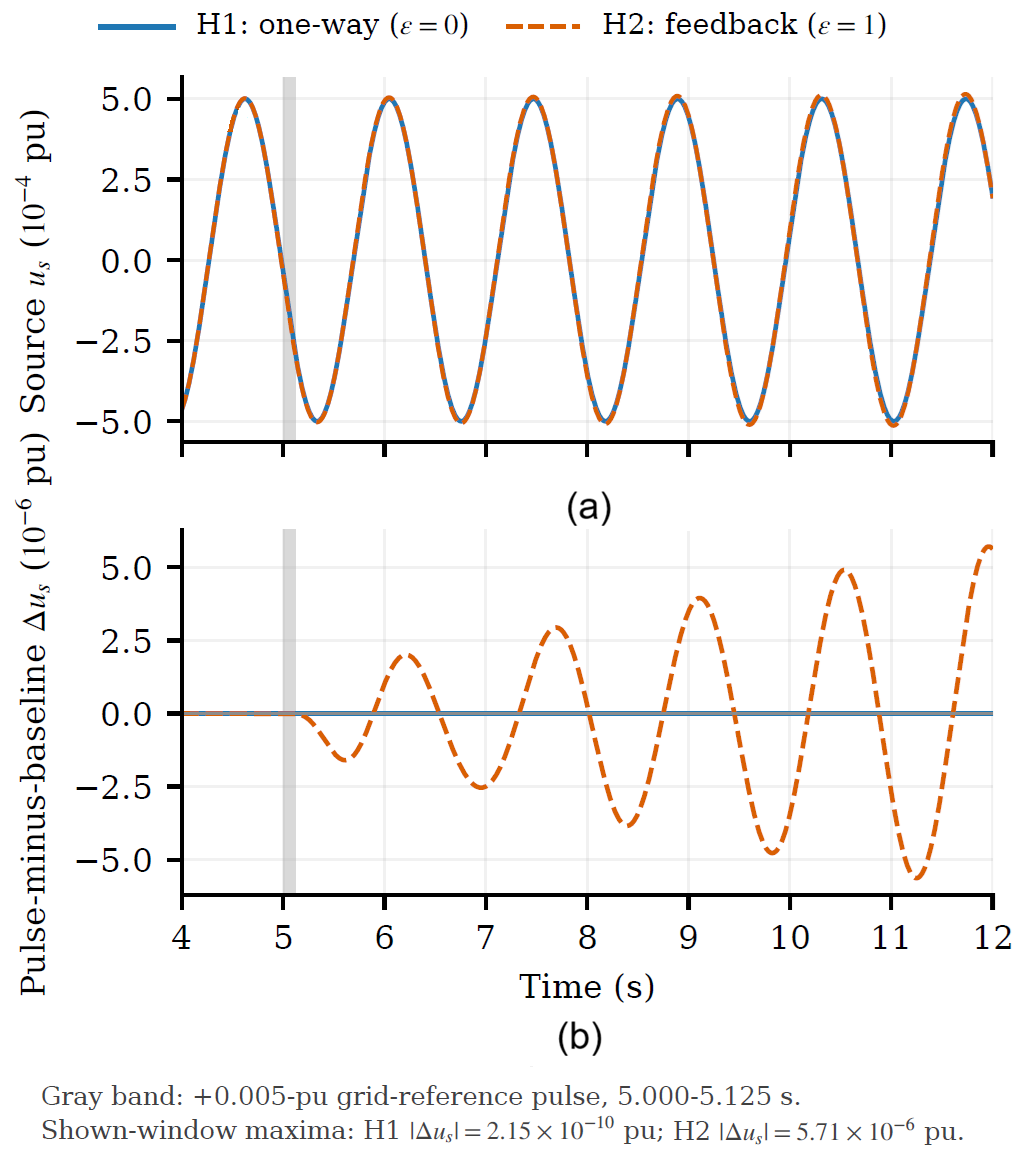}
\caption{Grid-reference intervention in the Kundur study.
(a) H1 and H2 source-interface trajectories during a G1 AVR-reference
pulse. (b) Corresponding pulse-induced responses. The pulse is applied
through the grid to verify the implemented grid-to-source feedback path.}
\vspace{-0.2in}
\label{fig:H2_pulse_intervention}
\end{figure}

\vspace{-0.1in}
\subsection{Q2: H1/H2 Distinction Through Feedback, 
Modal Interaction, and Detuning}

Consider H2 with nonzero grid-to-source feedback, $\eps>0$.
Using the calibrated linear source model of Subsection~VI-A, the study
first verifies the feedback path and then varies $\eps$ and
source--grid detuning to examine how bidirectional coupling changes
the augmented-system modes and their interaction. As shown in
Fig.~\ref{fig:H2_pulse_intervention}, a grid-reference pulse changes
the H2 source trajectory while leaving the one-way H1 source
unchanged within numerical accuracy. The principal coupling sweep includes 
$\eps=$0, 0.001, 0.003, 0.01, 0.03, 0.1, 0.3, 1, 3, 10, beyond which field-voltage limits may be reached during a 61~s simulation.

\begin{figure}
\centering
\includegraphics[width=0.65\columnwidth,height=1.5in]{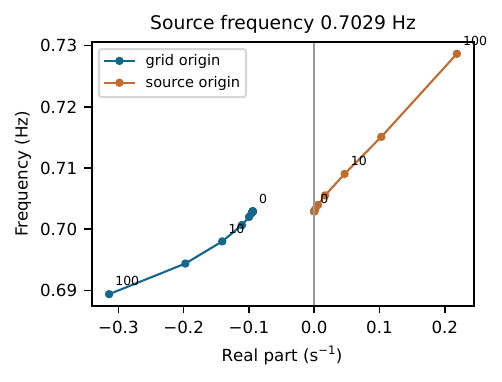}
\vspace{-0.2in}
\caption{Full augmented SSAT eigenvalue continuation as the feedback strength $\eps$ increases. The grid and source branches for the dominant inter-area mode move from their $\eps=0$ locations, illustrating source--grid modal interaction}
\vspace{-0.2in}
\label{fig:modal}
\end{figure}

The linearization of the augmented source-grid model via SSAT has 42 states,
compared with 40 states in the original system. 
Fig.~\ref{fig:modal} tracks
the eigenvalue branches originating from the dominant interarea
grid mode and the source oscillator as $\eps$ increases. At
$\eps=0$, the two branches retain their separate grid and source
origins. As feedback is increased, both branches move in the
complex plane: the grid-origin mode changes in both damping and
frequency, while the initially neutral source-origin mode moves
into the right half-plane.
This example is intended to illustrate feedback-induced eigenvalue migration
and modal interaction, rather than a stable self-sustained oscillation.
Stable interconnected modes, feedback-induced instability, and nonlinear
limit-cycle oscillations are distinct dynamical regimes; the latter is
examined in Section~VI-E using the Van der Pol source.
Thus bidirectional coupling changes
the natural modes of the augmented source--grid system rather
than merely changing the amplitude of a forced response. The
branch motion provides full-system numerical evidence of
source--grid modal interaction, consistent with the two-mode method of Section~IV.

To compare the full-system continuation with the reduced theory of
Section~IV, an effective loop product
$p=bc=\TwoModeProduct$ is fitted from the weak-coupling cases
$\eps=0.003,0.01,0.03$. Agreement at these points is therefore
calibration rather than independent validation. Beyond this fitted
range, the exact two-mode expression \eqref{eq:lpm} remains within the 5\% eigenvalue-shift criterion through approximately $\eps=\TwoModeRange$, providing an out-of-calibration comparison. The increasing error at stronger
coupling is consistent with the influence of modes omitted from the
two-mode reduction.

Detuning studies vary the source frequency around the dominant interarea
mode. Because the grid mode is damped, the relevant separation in
\eqref{eq:chi} is the complex modal distance
$|\lambda_g-\lambda_s|$, not frequency difference alone. With a fixed
physical noise level, moving the source frequency closer to $f_g$
increases the nominal response SNR from \DetuneFarSNR~dB to
\DetuneNearSNR~dB and the nominal feedback-detection probability from
\DetuneFarPower\% to \DetuneNearPower\% for the tested configuration.
Thus, resonance does not universally degrade feedback detection, but
this case does not establish that resonance generally improves
detectability; the result remains configuration dependent, and
frequency-specific numerical qualification is incomplete.

\vspace{-0.1in}
\subsection{Q2: H1/H2 Feedback Detection and Source Uncertainty}
This experiment tests whether finite noisy grid-side measurements favor the specified feedback-augmented family \(\mathcal M_{+}^{\rm cal}\) over its zero-feedback boundary \(\mathcal M_{0}^{\rm cal}\). A positive result provides conditional evidence for H2-type grid-to-source feedback within this assumed source family; it is not discrimination between arbitrary members of the structural classes H1 and H2. 
The linear source family in \eqref{eq:tested_source} is considered in three scenarios. In Scenario 1, $\eps$ is the only unknown parameter, with source frequency, amplitude, phase, damping, and feedback normalization fixed; Scenario 2 additionally treats source amplitude $A_s$ and phase $\phi_s$ as unknown, while Scenario 3 further includes source
frequency $f_s$ and damping $\zeta_s$. For parameter vector \(\theta=[\eps,\;A_s,\;\phi_s,\;f_s,\;\zeta_s]^{\mathsf T}
\), parameter sensitivities and the projected feedback sensitivity $r_{\eps}$ are defined as in \eqref{eq:projected_feedback_sensitivity}.  Scenario 1 provides the calibrated
feedback-detection study, whereas Scenarios 2 and 3 assess local joint parameter
resolution rather than additional detection trials.

\subsubsection{Scenario 1: feedback detection with only $\eps$ unknown} Clean simulation results from TSAT are downsampled to 30 and 60 samples/s with Gaussian measurement noise added at 20--60~dB SNR. Observation windows of 5--60~s and the three measurement configurations are considered. The
coupling strength is sampled at
$\eps\in\{0,0.001,0.003,0.01,0.03,0.1,0.3,1,3,10\}$, with 200 noise realizations for each case. 

In Scenario~1, $\eps$ is the only unknown parameter, so
$r_\eps=\bar s_\eps$ and
$\|r_\eps\|/\|\bar s_\eps\|=1$ by construction. For the stacked measurement vector $d$, the simulated responses at the
listed coupling values define a piecewise-linear interpolant $m(\eps)$
over $0\leq\eps\leq10$. The coupling parameter is estimated by weighted
least squares:
\begin{equation}
J(\eps)=\|d-m(\eps)\|_{\Sigma^{-1}}^2,\qquad
\widehat{\eps}=\operatorname*{arg\,min}_{0\leq\eps\leq10}J(\eps).
\label{eq:eps_est}
\end{equation}
Check $\Lambda=J(0)-J(\widehat{\eps}).$ Feedback is declared when $\Lambda$ exceeds the empirical 95th percentile
of 10,000 independent $\eps=0$ calibration trials. Repeating this test over 200 noise realizations gives the selection probability $P_{\rm sel}$ reported in Table~\ref{tab:modelselection}; its 95\% interval is the corresponding Wilson interval.

\begin{table}[!t]
\centering
\footnotesize
\caption{Scenario~1 results for $\Ylocal$ at 30 samples/s and 20~dB.
$P_{\mathrm{sel}}$ is the empirical probability that the calibrated \(\Lambda\) test selects positive feedback, i.e., \(\mathcal M_{+}^{\rm cal}\), over \(\mathcal M_{0}^{\rm cal}\); intervals are pointwise 95\%
Wilson intervals. Positive-feedback rows are limited to cases for which
the numerical-resolution screen was evaluated.}
\label{tab:modelselection}
\begin{tabular}{ccccc}
\hline
$T$ (s) & $\eps$ & $P_{\mathrm{sel}}$ & 95\% interval & Num. qual. \\
\hline
5  & 0    & 0.055 & [0.031,0.096] & N/A  \\
5  & 0.03 & 0.055 & [0.031,0.096] & Fail \\
5  & 1    & 0.085 & [0.054,0.132] & Fail \\
5  & 10   & 0.975 & [0.943,0.989] & Pass \\
60 & 0    & 0.050 & [0.027,0.090] & N/A  \\
60 & 0.03 & 0.965 & [0.930,0.983] & Fail \\
60 & 1    & 1.000 & [0.981,1.000] & Pass \\
60 & 10   & 1.000 & [0.981,1.000] & Pass \\
\hline
\end{tabular}
\end{table}

\begin{figure}[!t]
\centering
\vspace{-0.1in}
\includegraphics[width=\columnwidth, height=1in]{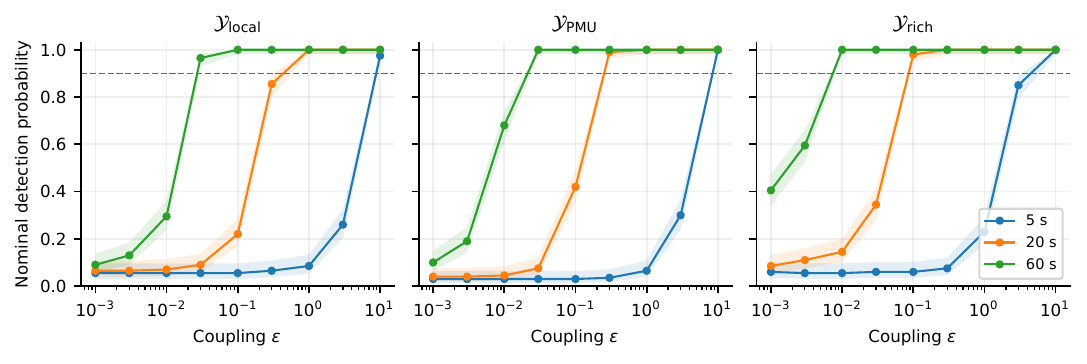}
\vspace{-0.3in}
\caption{Nominal feedback-detection probability for positive $\eps$ under
$\Ylocal$, $\YPMU$, and $\Yrich$ at 20~dB and 30 samples/s. Bands are
pointwise 95\% Wilson intervals. Numerical qualification is assessed
separately from the statistical detection result.}
\label{fig:detection}
\vspace{-0.2in}
\end{figure}

Table~\ref{tab:modelselection} and Fig.~\ref{fig:detection} summarize the
results. Numerical qualification requires the weighted H2--H1 contrast to change by no more than 5\%
between the simulations using $1/128$-s and $1/256$-s timesteps on both estimation and
held-out intervals, and to remain above the estimated
numerical/export resolution level so that the observed difference is not
attributable to solver or data-export precision. Hence, a large nominal detection probability does not by itself
establish physically resolved feedback. Additionally, for $\Ylocal$ at 30 samples/s and 20~dB, the short-window case at
$\eps=\LocalShortEpsilon$ gives
$P_{\mathrm{sel}}=\LocalShortPower\%$, but does not pass the numerical
screen. At $\eps=10$, the detection probability increases to 97.5\% and
the paired-contrast screen is satisfied. For a 60-s window at 40~dB,
the smallest numerically qualified tested coupling is
$\eps=\SmallestQualifiedCoupling$ for all three measurement
configurations. 


\subsubsection{Scenario 2: Unknown source amplitude and phase}

This scenario tests whether source-amplitude and phase variations can
absorb the feedback signature. At 60~s, 30 samples/s, and 40~dB, the
scaled condition numbers are 12.86, 13.55, and 15.15 for
$\Ylocal$, $\YPMU$, and $\Yrich$, respectively, and all three singular
directions exceed the coarse/fine finite-difference discrepancy. The
corresponding residual fractions $\|r_\eps\|/\|\bar s_\eps\|$ are
0.475, 0.466, and 0.457, indicating that the feedback direction remains
locally distinguishable to first order from the nuisance-parameter
sensitivity subspace.

\subsubsection{Scenario 3: Unknown feedback, amplitude, phase, frequency, and damping}
At the same setting, the scaled condition numbers increase to 266.2,
233.9, and 249.2 for $\Ylocal$, $\YPMU$, and $\Yrich$, respectively.
Four of five singular directions exceed the coarse/fine finite-difference
discrepancy, while one remains unresolved. The residual fractions decrease
to 0.0377, 0.0434, and 0.0408, indicating that the feedback direction is
nearly contained in the nuisance-parameter sensitivity subspace and is
therefore only weakly resolved separately.

\vspace{-0.1in}
\subsection{Nonlinear Sources and Adverse Cases}
\label{sec:kundur-adverse}

The preceding H2 studies assume the linear source model in \eqref{eq:tested_source}.
This subsection relaxes that assumption to examine whether the frequency alignment of the grid response with the source alone reveals feedback, 
to consider source-model mismatch, and to provide a
counterexample to a simple frequency-proximity classification rule.

\subsubsection{Synchronization with a nonlinear self-sustained source}

The source model in \eqref{eq:tested_source} is augmented with the
Van der Pol term $\mu(1-z_1^2)z_2$, with
$\mu=0.2,0.5,$ and $1~\mathrm{s}^{-1}$, and tested at detunings of
$-0.1$, $0$, and $+0.1$~Hz. The source-only cases are first verified to
exhibit stable self-sustained oscillations, after which the coupled system
is evaluated for $\eps=0,1,$ and $10$. All cases show close source--grid synchronization over the observation
window, including the one-way cases with $\eps=0$. Thus, a common
source--grid frequency or finite-window synchronization does not by itself
establish grid-to-source feedback; a one-way self-sustained source can
produce the same qualitative behavior.
Bidirectional feedback nevertheless changes the source dynamics. At zero
detuning and $\eps=10$, the source frequency shifts by 4.10--4.26~mHz
across the three nonlinear parameterizations, demonstrating frequency
pulling under feedback. These shifts are case-dependent, not a universal indicator of bidirectional coupling.

\subsubsection{Source-model mismatch}

A source-model mismatch test examines whether an incorrect source family
can produce a false indication of feedback. The true response is generated
by the nonlinear Van der Pol source, while inference uses the linear-source
model of the preceding studies.
For $\mu=1$, zero detuning, $\Ylocal$, 20~s, 30 samples/s, 20~dB, and
true $\eps=0$, \IDFiveSelection\% of 200 trials select
the feedback-augmented candidate family \(\mathcal M_{+}^{\rm cal}\), with mean fitted coupling \IDFiveEstimate. Selection of \(\mathcal M_{+}^{\rm cal}\) alone does not establish physical H2 feedback.
The misspecified linear model also gives a larger held-out weighted
mean-square error, \IDFiveMSE, than the correctly specified nonlinear
oracle model, \IDFiveOracleMSE. Thus, nonlinear source-model mismatch can
be misinterpreted as grid-to-source feedback, so selection of
$\mathcal{M}_2$ alone does not establish physical feedback.

\subsubsection{Adverse frequency-only classification}
Finally, a frequency-proximity rule classifies an oscillation as
``natural'' when $|\hat f-f_g|<0.02$~Hz. For a near-resonant
one-way-equivalent case, the rule classifies \FrequencyFalseNatural\%
of 200 noisy trials as ``natural.'' Thus, proximity to a grid-mode
frequency does not establish bidirectional source--grid dynamics; a
one-way source can produce a mode-like response near resonance.

Together, these cases show that synchronization, frequency proximity,
and improved fit of an assumed feedback model are insufficient by
themselves to establish grid-to-source feedback. Reliable inference
requires an appropriate source model, numerically resolved feedback
signatures, and informative measurements.

\vspace{-0.1in}
\subsection{Synthesis of the Kundur System Results}
\label{sec:kundur-synthesis}

\begin{figure}[!t]
\centering
\includegraphics[width=\columnwidth,height=1.9in]
{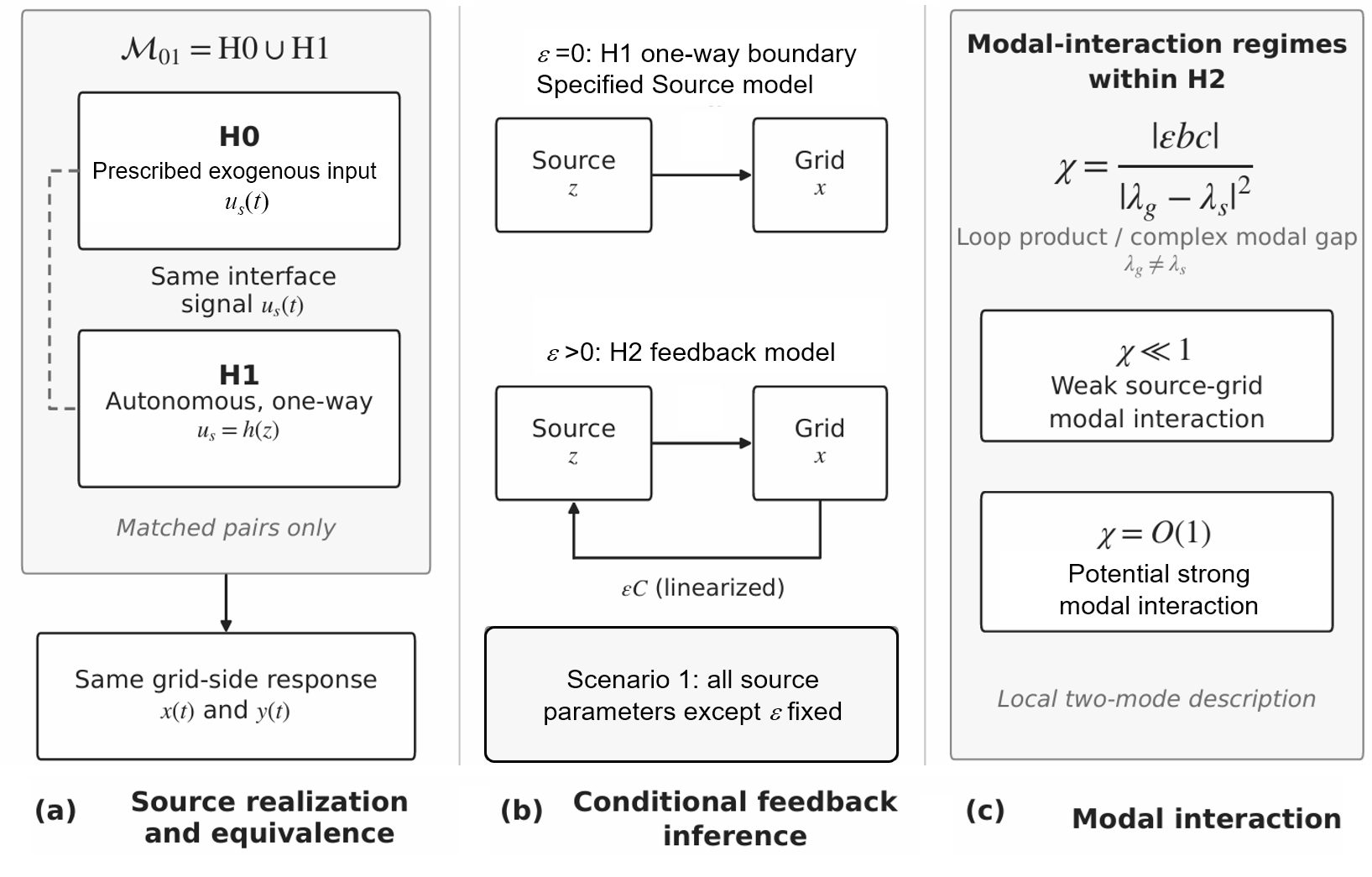}
\caption{Conceptual synthesis of the Kundur studies:
(a) observational equivalence of matched H0/H1 realizations;
(b) conditional detection of H2 feedback; and
(c) source--grid modal interaction under bidirectional coupling.}
\vspace{-0.2in}
\label{fig:hypothesis-concept-synthesis}
\end{figure}

Fig.~\ref{fig:hypothesis-concept-synthesis} summarizes the three main
findings. Matched H0 and H1 realizations can be observationally
equivalent from grid-side measurements; H2 feedback can be distinguished
from the one-way boundary only under suitable source-model, measurement,
and numerical-resolution conditions; and once bidirectional coupling is
present, source and grid modes can interact. These results define a
dynamical boundary between one-way and bidirectionally coupled
oscillations without implying a universal forced-versus-natural
classifier based on passive grid measurements alone.

\section{Large-System Test}
\label{subsec:naspi_controlled}

A modification of Case~7 in the 2021 IEEE--NASPI OSL Contest
is used to demonstrate the proposed approach for answering questions Q1 (H0/H1 equivalence) and Q2 (H1/H2 feedback detection) on the
243-bus, 140-generator WECC testbed with 2139 dynamic states
\cite{MaslennikovWang2022Contest,Maslennikov2016Library,Wang2025PMUOSL}.
A 0.379-Hz source is applied through the exciter of generator 2634/C as
$u_s=k_z z_1$, with $k_z=0.02$~pu, and is governed by
\begin{equation}
 \dot z_1=z_2,\qquad
 \dot z_2=-\omega_s^2z_1+\eps K_{\rm ref}\mathrm{DW},
 \label{eq:naspi_source}
\end{equation}
where $\omega_s=2\pi(0.379)$~rad/s,
$K_{\rm ref}=56707.19$~s$^{-2}$, and DW is the source-generator
rotor-speed deviation in pu.

\subsection{Q1: H0/H1 Observational Equivalence}

The H0 and H1 cases use the same source interface, with H1 remaining
one-way coupled to the grid. Their exported source-injection waveforms
match at numerical precision, with only small discrepancies  of  less than $10^{-6}$~pu in voltage magnitude and less than $10^{-3}$ degree angle. This is consistent with Theorem~\ref{thm:interface} and
Corollary~\ref{cor:h0h1}: matched H0 and H1 realizations are
observationally equivalent from grid-side measurements.

\subsection{Q2: H1/H2 Feedback Distinguishability}

Scenario~1 is tested here on the zero-feedback boundary against H2 using a fixed source model and grid-side measurements. Inference is restricted to
$\eps\in\{0,0.01,0.02,0.03,0.06,0.10\}$, with all other source
parameters fixed. Unlike the Kundur study, these
six directly simulated candidates are compared without interpolation or
continuous optimization. The primary case uses bus-2604 voltage magnitude and
angle at 30 samples/s over 60~s, with 20-dB measurement noise and an
80/20 estimation--validation split. Feedback is detected using the
calibrated improvement statistic $\Lambda$ in \eqref{eq:Lambda}, with its threshold determined from independent $\eps=0$ trials. 
The restricted zero-feedback family \(\mathcal M_{0}^{\rm cal}\) is tested against the positive-feedback candidate family \(\mathcal M_{+}^{\rm cal}\). For true \(\epsilon=0.02\), positive feedback, i.e. \(\mathcal M_{+}^{\rm cal}\), is selected in all 500 trials, and
the mean held-out weighted MSE decreases from 2.1001 for \(\mathcal M_0^{\rm cal}\) to 1.0021 for the selected positive-feedback candidate. Statistical detection and numerical validation are treated separately.
The same numerical-qualification rule is applied using the simulations at $1/256$-s
and $1/512$-s timesteps, with a third-step refinement check where
available. The local 60-s cases at
$\eps=0.02,0.06,$ and $0.10$ satisfy these criteria; the other positive
detections do not pass the numerical-validation screen and are therefore
not interpreted as numerically verified feedback signatures.

Feedback detection should not be confused with parameter recovery. For example, at true $\eps=0.10$, feedback is detected in all 500
10-s trials, while the exact $\eps=0.10$ candidate is selected in
476 cases, with the remaining 24 picking $\eps=0.06$.

\begin{figure}[t]
\centering
\includegraphics[width=\columnwidth,height=1.1in]
{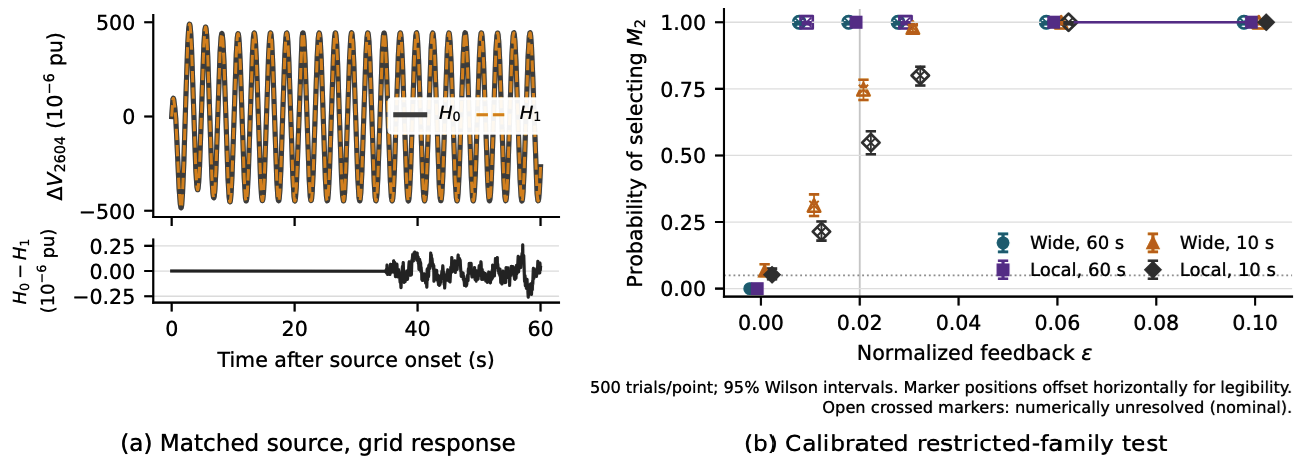}
\caption{IEEE-NASPI Case~7.
(a) Bus-2604 H0/H1 voltage-magnitude responses and their signed
difference, with time measured from source activation.
(b) Calibrated feedback-detection rates with pointwise 95\% Wilson
intervals. Crosses indicate cases that fail the numerical-validation
screen.}
\label{fig:naspi_controlled}
\vspace{-0.1in}
\end{figure}

\vspace{-0.1in}
\section{Conclusions}

This paper establishes source--grid coupling directionality as a useful
boundary for interpreting sustained oscillations. A prescribed forcing
signal and an autonomous one-way source can be observationally equivalent
from grid-side measurements when they produce the same interface signal.
Conversely, frequency proximity to a grid mode, source localization, or
frequency alignment of the grid response with the source does not by itself establish bidirectional
coupling. The proposed framework therefore separates identification of
an oscillatory component from the distinct question of whether the grid
dynamically feeds back into its source.
The Kundur and WECC system studies show that grid-to-source feedback can
be detected from grid-side measurements under specified source models
and sufficient numerical and statistical resolution. However, feedback
detection does not require unique recovery of the source parameters, and
source-model mismatch can be misinterpreted as grid-to-source feedback.
Likewise, failure to detect feedback does not establish that the feedback
is absent. Future work will study feedback distinguishability under joint source-parameter
uncertainty and model mismatch, and the additional measurements or interventions
needed to resolve remaining ambiguities.

\vspace{-0.1in}
\section*{Acknowledgment}
The authors used OpenAI ChatGPT and Codex to assist with editing and technical exposition in Sections II--VIII and with selected computational
workflows supporting the numerical studies in Sections VI and VII.
All mathematical developments, simulation results, interpretations, and
conclusions were reviewed and verified by the authors.
\vspace{-0.15in}

\bibliographystyle{IEEEtran}
\bibliography{references}

\end{document}